\documentclass[a4paper,11pt]{article}

\usepackage{graphicx}

\usepackage{fullpage}
\usepackage{libertine}
\usepackage{color}
\usepackage{subcaption}
\usepackage{multirow}

\usepackage[ruled,vlined]{algorithm2e}
\SetArgSty{textrm}

\usepackage{amsmath,amsfonts,amssymb,amsthm}

\usepackage[breaklinks=true]{hyperref}
\usepackage[svgnames]{xcolor}
\usepackage[capitalise,nameinlink]{cleveref}
\hypersetup{colorlinks={true},linkcolor={DarkBlue},citecolor=[named]{DarkGreen}}

\usepackage{natbib}
\usepackage{authblk}

\usepackage{tikz}  
\usetikzlibrary{arrows}
\usetikzlibrary{patterns,snakes}
\usetikzlibrary{decorations.shapes}
\tikzstyle{overbrace text style}=[font=\tiny, above, pos=.5, yshift=5pt]
\tikzstyle{overbrace style}=[decorate,decoration={brace,raise=5pt,amplitude=3pt}]
\usetikzlibrary{shapes.geometric}
\usepackage{pgfplots}
\pgfplotsset{compat=1.16}
\newtheorem{theorem}{Theorem}[section]

\newtheorem{lemma}[theorem]{Lemma}

\theoremstyle{definition}
\newtheorem*{comment*}{Comment}

\newcommand{\cost}{\text{cost}}

\usepackage{mathtools}

\DeclareMathOperator*{\argmax}{arg\,max}

\renewcommand{\top}{\mathtt{top}}

\title{\bf Group-Fair Metric Distortion of Facility Assignment Problems}
\author{Alexandros A. Voudouris}

\affil{Department of Mathematics and Computer Science \\University of Southern Denmark, Denmark}

\date{}

\begin{document}

\allowdisplaybreaks

\maketitle

\begin{abstract}
We study the group-fair distortion of metric facility assignment problems, where a set of agents, partitioned into unknown groups, must be assigned to a collection of facilities, possibly subject to capacity or other feasibility constraints. Given an assignment, each agent incurs a cost that depends on both its distance to its assigned facility and, via an affinity factor, the average distance of the other members in its group to their assigned facilities. We consider full-information algorithms, which have complete knowledge of the metric space, and ordinal-information algorithms, which know the distances between facilities and only the rankings of the agents over facilities (sorted by increasing distance). We establish worst-case distortion upper bounds in terms of the Max-of-Sum and Sum-of-Max social objectives, which combine the classic utilitarian and egalitarian social cost measures. We also derive informational lower bounds for one-sided matching and clustering, two fundamental and well-studied problems captured by our model, that match our upper bounds exactly for Max-of-Sum and asymptotically for Sum-of-Max.

\medskip
\noindent 
{\bf Keywords:} Distortion; Group-fair objectives; Facility assignment
\end{abstract}

\section{Introduction} \label{sec:intro}

We consider facility assignment problems with a set $N$ of $n \geq 2$ {\em agents} and a set $A$ of $m\geq 2$ {\em facilities}. Agents and facilities are located in a metric space with distance function $d:(N \cup A) \times (N \cup A) \to \mathbb{R}_{\geq 0}$. Thus, all distances are nonnegative and satisfy the triangle inequality: $d(x,y) \leq d(x,z) + d(z,y)$ for any $x,y,z \in N \cup A$. The agents are also partitioned into a set $G$ of $k \geq 2$ unknown {\em groups}. We will use the tuple $I = (N,A,G)$ to denote an instance of the problem. Our goal is to compute an {\em assignment} $X = (x_i)_{i\in N}$ of agents to facilities (agent $i$ is assigned to facility $x_i$) to minimize a social objective that is a function of the individual costs of the agents for the assignment. 
We pay particular attention to two fundamental and well-studied assignment problems of this kind, each with its own constraints on the feasible assignments:
\begin{itemize}
    \item {\bf One-sided Matching:} There are $m=n$ facilities, and the assignment must be a one-to-one matching between agents and facilities~\citep{hylland1979efficient,house}. 
    
    \item {\bf Clustering:} $\lambda < m$ of the $m$ available facilities are chosen as {\em cluster centers} or {\em centroids}, and each agent must be assigned to its closest cluster center~\citep{SCHAEFFER200727,An2017}. 
\end{itemize}

\medskip
\noindent
{\bf Individual cost and social objectives.}
Given an assignment $X = (x_i)_{i \in N}$, each agent $i$ suffers an intragroup {\em individual cost} that depends to some extent on the distances of all members of $i$'s group from their assigned facilities. In particular, for an affinity parameter $\alpha \in [0,1]$, which specifies the extent to which the whole group affects the cost of each of its members and aims to capture how the happiness of one group member affects the others, we define the individual cost of an agent $i$ that belongs to a group $g \in G$ of size $n_g \geq 2$ as 
\begin{align*}
    \cost_i(X) = d(i,x_i) + \frac{\alpha}{n_g-1}\sum_{j \in g\setminus\{i\}} d(j,x_j). 
\end{align*}
In words, the individual cost of agent $i$ is its own distance from its assigned facility $x_i$, plus an $\alpha$-factor of the average distance of the other members in $g$ from their assigned facilities. 

The social cost of the overall assignment $X$ is a function of the individual agent costs. We consider the following two social objectives, which are well-defined for the group setting we study:
\begin{itemize}
    \item {\bf Max-of-Sum:} The maximum over all groups of the total individual cost of the agents therein: 
    \[ \text{Max-of-Sum}(X) = \max_{g \in G} \sum_{i \in g} \cost_i(X).\]
    \item {\bf Sum-of-Max:} The total over all groups of the maximum individual cost of the agents therein: 
    \[ \text{Sum-of-Max}(X) = \sum_{g \in G} \max_{i \in g} \cost_i(X).\]
\end{itemize}
Variants of these objectives (taking the average instead of the sum) were first proposed in the context of distributed metric voting by \citet{AnshelevichFV2022distributed}, but can be more generally thought of as aiming to capture different notions of fairness between groups. As such, they were recently considered by \citet{AmanatidisAJV2025group} to study the group-fair distortion of single-winner voting. Observe that, to optimize the Max-of-Sum objective, we need to compute an assignment that balances, over all groups, the total individual cost within each group. On the other hand, to optimize the Sum-of-Max objective, we need first to control the maximum individual cost within each group. When the social objective is clear from context, we will simply use $\cost(X)$ to denote the social cost of assignment $X$ according to the objective. 

\medskip
\noindent 
{\bf Types of information and distortion.}
Following the work of \citet{AnshelevichZ21}, we make the natural assumption that the distances between facilities in the underlying metric space are known, and we have some additional knowledge about the distances between agents and facilities. In particular, we consider two different types of information:
\begin{itemize}
    \item {\bf Full information:} We know the exact distances of the agents from the facilities, and we thus have complete knowledge of the metric space.  
    
    \item {\bf Ordinal information:} For each agent $i$, we are given a ranking $\succ_i$ of the facilities sorted from smallest to largest distance, that is $x \succ_i y$ implies that $d(i,x) \leq d(i,y)$. The ranking $\succ_i$ is also known as the {\em ordinal preference} of agent $i$ over the facilities.
\end{itemize}
We will refer to algorithms that have complete access to the metric space as {\em full-information algorithms}, and algorithms that have access to the ordinal preferences of the agents over the facilities as {\em ordinal-information algorithms}. For both classes, either because of the unknown group structure (when full information is available), or because of the unknown group structure and the limited information about the distances of the agents from the facilities, some loss of efficiency in terms of the aforementioned social objectives is expected. The {\em distortion} of an algorithm $\mathcal{A}$ for a facility assignment problem is the worst-case (over all possible instances with $|N|=n$ agents, $|A|=m$ facilities, and $|G|=k$ groups) ratio between the social cost of the assignment computed by $\mathcal{A}$ and the minimum possible social cost over all possible assignments:
\begin{align*}
    \sup_{I=(N,A,G)} \frac{\cost(\mathcal{A}(I))}{\min_X \cost(X)}.
\end{align*}
By definition, the distortion is always at least $1$, and our goal is to design algorithms that achieve as low a distortion as possible.

\subsection{Our Contribution} \label{sec:contribution}

We show upper bounds on the distortion of full-information and ordinal-information algorithms for general facility assignment problems with respect to Max-of-Sum and Sum-of-Max. In particular, for Max-of-Sum, using as a black box any $\rho$-approximation algorithm for the SUM problem of computing an assignment with minimum total distance over all agents, we derive an upper bound of $\rho k$ on the distortion of full-information algorithms and an upper bound of $2\rho k+1$ on the distortion of ordinal-information algorithms.
Similarly, for Sum-of-Max, using any $\rho$-approximation algorithm for the MAX problem of computing an assignment that minimizes the maximum distance over all agents, we show upper bounds of $(1+\alpha) \rho k$ and $(1+\alpha)(2\rho k+1)$ for full-information and ordinal-information algorithms, respectively. We complement these general results with informational lower bounds for one-sided matching and clustering that hold for all algorithms, even those with infinite computational resources. These bounds match the corresponding upper bounds exactly for Max-of-Sum and asymptotically in $k$ for Sum-of-Max.

\subsection{Related Work} \label{sec:related}

The distortion framework~\citep{distortion-survey} has become central for quantifying the inefficiency of mechanisms that operate under limited information in a plethora of utilitarian social choice settings, including {\em single-winner voting}~\citep{Anshelevich2018approximating,boutilier2015optimal,charikar24breaking,gkatzelis2020resolving,kempe2022veto}, {\em multi-winner voting}~\citep{caragiannis2017subset,caragiannis2022multi}, {\em matching}~\citep{amanatidis2022matching,anari2023matching,caragiannis2024augmentation}, and {\em clustering}~\citep{BurkhardtCFRSS24}. An assumption made implicitly in all the models considered in the aforementioned works is that the agents form a single pool, the mechanism has access to all their preferences at once, and the goal is to optimize a centralized objective function, such as the social or egalitarian cost in the metric setting, typically defined as the total or maximum individual cost of the agents, respectively. 

A recent stream of work has studied distributed settings, where the agents are partitioned into disjoint groups (corresponding, e.g., to electoral districts) that make independent proposals, which are then used to compute the final outcome~\citep{Filos-RatsikasM20}. In metric models, the efficiency of the outcome can be measured in terms of the social and egalitarian cost, but also in terms of tailor-made group-fair objectives that are functions of the groups, such as variants of Max-of-Sum and Sum-of-Max, which combine in some form the social and egalitarian costs within and across the groups~\citep{AnshelevichFV2022distributed,Voudouris2023tight,FilosRatsikasKVZ24,Xu2025distributed,Abam2025randomized-distributed}.

The work most closely related to ours is that of \citet{AmanatidisAJV2025group}, who considered a metric single-winner group voting model that lies between the centralized and distributed settings discussed above, where agents are partitioned into groups, but the outcome is centrally computed using the preferences of all agents. In this setting, the main source of inefficiency in terms of variants of group-fair objectives (like Max-of-Sum and Sum-of-Max) is the possibly unknown structure of the groups of agents, as well as the limited information used to describe their preferences. Similar group models have been studied in the context of truthful facility location, where the goal is to design strategyproof mechanisms that simultaneously achieve a good approximation in terms of group-fair objectives~\citep{ZhouLC2022group,LiLC2024group-obnoxious,LiPWZ2025group,WangZL24}.

As already mentioned, matching problems, including one-sided matching, have been studied within the distortion framework under different informational assumptions. Most notably, our work follows that of \citet{AnshelevichZ21}, who make the natural assumption that the facility locations in the metric space are known. However, in some applications, where the facilities do not correspond to physical buildings but to items or resources, we might not have any information about their distances. The distortion of metric one-sided matching when only ordinal information about the preferences of the agents over the items is given has recently been studied in a series of papers~\citep{caragiannis2024augmentation,anari2023matching,Filos-RatsikasG25,hastings2025fairmetricdistortionmatching} and is still a wide-open problem. 

Clustering problems such as $k$-median, $k$-center, and $k$-means, where the full metric space is given, have been thoroughly studied within the optimization literature, and approximation algorithms for these NP-hard problems with small constant bounds are known~\citep{CharikarGTS02,LiS16,ByrkaPRST17,AhmadianNSW20,Cohen-Addad0LSS25}. Ordinal versions, where agents rank the possible cluster centers by distance, can be thought of as variants of multiwinner voting with ordinal preferences~\citep{caragiannis2022multi}.  
More closely related to our work is the clustering literature on socially fair clustering, where the agents are partitioned into groups and the objective is to balance the cluster costs across groups by optimizing functions similar to Max-of-Sum~\citep{MakarychevV21,AbbasiBV21,GhadiriSV21,DickersonEMZ25}.

\section{Max-of-Sum} \label{sec:max-of-sum}
We start with the Max-of-Sum social objective. 
To upper-bound the distortion of algorithms in terms of this objective, we will need to calculate the total individual cost of agents within a group.

\begin{lemma} \label{lem:group-cost}
For any assignment $X$, the total individual cost of the agents in a group $g$ is 
\[\cost_g(X) = \sum_{i \in g} \cost_i(X) = (1+\alpha) \sum_{i \in g} d(i,x_i).\]
\end{lemma}

\begin{proof}
By definition, the individual cost of an agent $i$ that belongs to group $g$ is $\cost_i(X) = d(i,x_i) + \frac{\alpha}{n_g-1} \sum_{j \in g\setminus\{i\}} d(j,x_j)$. Consequently, the total cost of the agents in $g$ is
\begin{align*}
    \cost_g(X) 
    &= \sum_{i \in g} \cost_i(X)\\
    &= \sum_{i \in g} \bigg( d(i,x_i) + \frac{\alpha}{n_g-1} \sum_{j \in g\setminus\{i\}} d(j,x_j) \bigg) \\
    &= \sum_{i \in g} \bigg( \left(1-\frac{\alpha}{n_g-1}\right)d(i,x_i) + \frac{\alpha}{n_g-1} \sum_{j \in g} d(j,x_j) \bigg) \\
    &= \left(1-\frac{\alpha}{n_g-1}\right) \sum_{i \in g} d(i,x_i) + \frac{\alpha \cdot n_g}{n_g-1} \sum_{i \in g} d(i,x_i)  \\
    &= (1+\alpha) \sum_{i \in g} d(i,x_i),
\end{align*}
as desired.
\end{proof}

\subsection{Full-Information Algorithms} \label{sec:max-of-sum:full}
Since the entire metric space is known, we can employ known approximation algorithms that use all of this information. In particular, we consider algorithms that are $\rho$-approximate for the {\em SUM problem} of minimizing the total distance of the agents from their assigned facilities.\footnote{For one-sided matching, the SUM problem is equivalent to computing a minimum-weight matching, while for clustering, it is equivalent to $\lambda$-median.}
That is, such an algorithm outputs an assignment $W$ which, for any assignment $X$, satisfies the inequality
\begin{align*}
    \sum_{i \in N} d(i,w_i) \leq \rho \cdot \sum_{i \in N} d(i,x_i). 
\end{align*}
We show that such an assignment achieves a distortion of $\rho k$ for Max-of-Sum. 

\begin{theorem} \label{thm:max-of-sum:full:upper}
Any $\rho$-approximate algorithm for the SUM problem achieves a distortion of at most $\rho k$ in terms of the Max-of-Sum objective. 
\end{theorem}

\begin{proof}
Let $W$ be the assignment computed by the algorithm, and let $O$ be an optimal assignment for the Max-of-Sum objective.
Since the algorithm is $\rho$-approximate for the SUM problem, we have 
\begin{align} 
    \sum_{i \in N} d(i,w_i) \leq \rho \cdot \sum_{i \in N} d(i,o_i). \label{eq:max-of-sum:full:rho-apx-SUM}
\end{align}
For the optimal assignment $O$, we have that $\cost(O) \geq \cost_g(O)$ for any group $g$. 
By summing these $k$ inequalities (one for each group) and using Lemma~\ref{lem:group-cost}, we obtain 
\begin{align}
    k \cdot \cost(O) 
    &\geq \sum_{g \in G} \cost_g(O) 
    = (1+\alpha) \sum_{g \in G} \sum_{i \in g} d(i,o_i)
    = (1+\alpha) \sum_{i \in N} d(i,o_i). \label{eq:max-of-sum:full:optimal-bound}
\end{align}
Now, let $g$ be the group that determines the Max-of-Sum cost of $W$. By Lemma~\ref{lem:group-cost}, \eqref{eq:max-of-sum:full:rho-apx-SUM} and \eqref{eq:max-of-sum:full:optimal-bound}, we have
\begin{align*}
    \cost(W) = \cost_g(W) 
    &= (1+\alpha) \sum_{i \in g} d(i,w_i) \\
    &\leq (1+\alpha) \sum_{i \in N} d(i,w_i) \\
    &\leq \rho \cdot (1+\alpha)\sum_{i \in N} d(i,o_i) \\
    &\leq \rho k\cdot \cost(O).
\end{align*}
The proof is complete.
\end{proof}

Next, we show a lower bound of $k$ on the distortion of full-information algorithms for one-sided matching and clustering. Together with the above positive result, this establishes a tight bound of $k$ for these two facility assignment problems when the underlying SUM problem is solved optimally, which can be done in polynomial time for one-sided matching but not for clustering. For clustering, this lower bound can be viewed as the burden of not knowing the group structure, even when unlimited computational resources are available to solve the centralized SUM problem. 

\begin{theorem}
For the Max-of-Sum objective, the distortion of any full-information one-sided matching algorithm is at least $k$.
\end{theorem}

\begin{proof}
Consider the following instance on a line metric:
\begin{itemize}
    \item There are $n=k^2$ agents located at $0$.
    \item There are $k(k-1)$ facilities located at $0$ and $k$ facilities located at $1$.
\end{itemize}
Let $W$ be the assignment computed by the algorithm on this metric space, and denote by $S$ the set of $k$ agents that are assigned to the facilities at $1$. The unknown groups might be such that
\begin{itemize}
    \item the $k$ agents of $S$ form a group;
    \item the remaining $(k-1)\cdot k$ agents are partitioned arbitrarily into $k-1$ groups of size $k$ each.
\end{itemize} 
Then, the Max-of-Sum cost of the computed matching is realized by the agents in $S$. Since the cost of each such agent is $1+\frac{\alpha}{k-1}(k-1) = 1+\alpha$, the overall cost is 
\[\cost(W) = k\left( 1+ \alpha \right).\]
We can create a matching $O$ such that exactly one agent from each group is assigned to a facility at $1$. Then, in each group, there is an agent with cost $1$ and $k-1$ agents with cost $\frac{\alpha}{k-1}$, leading to an optimal Max-of-Sum cost of 
\[\cost(O) = 1 + (k-1) \frac{\alpha}{k-1} = 1+\alpha,\] 
and a distortion of $k$.
\end{proof}

\begin{theorem}
For the Max-of-Sum objective, the distortion of any full-information clustering algorithm is at least $k$.
\end{theorem}

\begin{proof}
Consider the following instance on a line metric:
\begin{itemize}
    \item There are three facilities at $0$, $1$ and $2$.
    \item There are $k$ agents at $0$, $(k-1)k$ agents at $1$, and $k$ agents at $2$.
\end{itemize}
We switch between the possible locations of the cluster centers chosen by the algorithm; recall that each agent is assigned to its closest cluster center.

\medskip
\noindent
{\bf Case 1:} The cluster centers are $(0,2)$.
The $k$ groups might be as follows:
\begin{itemize}
\item The first group consists of the $k$ agents at $0$ and one agent at $2$.
\item Each of the remaining $k-1$ groups consists of $k$ agents at $1$ and one agent at $2$.
\end{itemize}
Then, the Max-of-Sum cost of the computed clustering $W$ is realized by any of the last $k-1$ groups.
In each such group, the agents at $1$ have cost $1+\frac{\alpha}{k}(k-1)$ and the agent at $2$ has cost $\frac{\alpha}{k}k = \alpha$. 
Hence,
\[\cost(W) = k\left(1+\frac{k-1}{k}\alpha\right) + \alpha = k(1 + \alpha).\] 
By choosing $(0,1)$ for the cluster centers, we have a clustering $O$ such that, for each group, the agent at $2$ has cost $1$ and each of the remaining $k$ agents (either at $0$ in the first group, or at $1$ in any of the last $k-1$ groups) has cost $\frac{\alpha}{k}$. So, the optimal Max-of-Sum cost is
\[\cost(O) = 1+k\cdot \frac{\alpha}{k} = 1+\alpha,\]
and the distortion is at least $k$ in this case. 

\medskip
\noindent
{\bf Case 2:} The cluster centers are $(0,1)$ (the case of $(1,2)$ is symmetric).
The $k$ groups might be as follows:
\begin{itemize}
    \item The first group consists of the $k$ agents at $2$ and one agent at $0$.
    \item Each of the remaining $k-1$ groups consists of $k$ agents at $1$ and one agent at $0$.
\end{itemize}
The Max-of-Sum cost of the computed clustering $W$ is realized by the first group. 
Each of the $k$ agents at $2$ has a cost of $1+\frac{\alpha}{k}(k-1)$ and the agent at $0$ has cost $\frac{\alpha}{k}k=\alpha$. 
Hence, the Max-of-Sum cost is again 
\[\cost(W) = k\left(1+\frac{k-1}{k}\alpha\right)+\alpha = k(1+\alpha).\]
By choosing $(1,2)$, we have a clustering $O$ such that, for each group, the agent at $0$ has cost $1$ and each of the $k$ remaining agents (either at $2$ in the first group, or at $1$ in any of the last $k-1$ groups) has cost $\frac{\alpha}{k}$, for a total of 
\[\cost(O) = 1+k\cdot\frac{\alpha}{k} = 1+\alpha.\]
Hence, the distortion is again at least $k$.
\end{proof}

\subsection{Ordinal-Information Algorithms} \label{sec:max-of-sum:ordinal}
We now turn our attention to the class of ordinal-information algorithms for the Max-of-Sum objective. Recall that such algorithms have access to the distances between the facilities in the metric space and to the ordinal preferences of the agents over the facilities. As with full-information algorithms, we again use a $\rho$-approximate algorithm for the SUM problem, but since we do not have complete knowledge of the metric space, we form a proxy instance in which each agent $i$ is located at its top-ranked (i.e., closest) facility $\top(i)$. Let $\widehat{W}=(\widehat{w}_i)_{i\in N}$ denote the assignment computed by the approximation algorithm in this proxy instance.

For clustering, the proxy algorithm selects a set $C_W$ of cluster centers, and $\widehat{w}_i$ is a closest center in $C_W$ to $\top(i)$. In the true instance, the ordinal algorithm assigns agent $i$ to its highest-ranked center $w_i$ in $C_W$, which is exactly a closest selected center to $i$. For all other facility assignment problems, the algorithm simply uses the proxy assignment, so $w_i=\widehat{w}_i$. Consequently, in all cases,
\begin{align} \label{eq:minitop-sum:transfer}
    d(i,w_i) \leq d(i,\top(i)) + d(\top(i),\widehat{w}_i).
\end{align}
Our ordinal algorithm, named {\sc $\rho$-MiniTopSUM}, therefore outputs the assignment $W=(w_i)_{i\in N}$, while the proxy assignment $\widehat{W}$ satisfies, for every feasible comparison assignment $X=(x_i)_{i\in N}$,
\begin{align} \label{eq:minitop-sum:proxy}
   \sum_{i \in N} d(\top(i),\widehat{w}_i) \leq \rho \cdot \sum_{i \in N} d(\top(i),x_i).
\end{align}
For clustering, Inequality~\eqref{eq:minitop-sum:proxy} follows by comparing the selected center set with the center set inducing $X$: assigning each proxy agent to its closest selected center can only decrease its distance.

\begin{theorem} \label{thm:max-of-sum:ordinal:upper}
Given any $\rho$-approximate algorithm for the SUM problem, the distortion of {\sc $\rho$-MiniTopSUM} is at most $2\rho k+1$ in terms of the Max-of-Sum objective. 
\end{theorem}

\begin{proof}
The proof follows along the lines of the proof of Theorem~\ref{thm:max-of-sum:full:upper}. Let $W$ be the assignment computed by the algorithm, and let $O$ be an optimal assignment. To lower-bound the optimal cost, we can again use Inequality~\eqref{eq:max-of-sum:full:optimal-bound} stating that
\begin{align*}
    k \cdot \cost(O) \geq (1+\alpha) \sum_{i \in N} d(i,o_i).
\end{align*}
Now, let $g$ be the group that determines the cost of $W$. Using Lemma~\ref{lem:group-cost}, the transfer bound in Inequality~\eqref{eq:minitop-sum:transfer}, the $\rho$-approximation guarantee of the proxy assignment in Inequality~\eqref{eq:minitop-sum:proxy}, the fact that $d(i,\top(i)) \leq d(i,o_i)$, and the above bound on the optimal cost (Inequality~\eqref{eq:max-of-sum:full:optimal-bound}), we get
\begin{align*}
\cost(W) 
    &= \cost_g(W) \\ 
    &= (1+\alpha) \sum_{i \in g} d(i,w_i) \\
    &\leq (1+\alpha) \sum_{i \in g} d(i,\top(i)) +  (1+\alpha) \sum_{i \in g} d(\top(i),\widehat{w}_i) \\
    &\leq (1+\alpha) \sum_{i \in g} d(i,o_i) + (1+\alpha) \cdot \rho \sum_{i \in N} d(\top(i),o_i) \\
    &\leq \cost_g(O) + (1+\alpha) \cdot \rho \sum_{i \in N} \bigg( d(i,\top(i)) + d(i,o_i) \bigg) \\
    &\leq \cost(O) + 2\rho \cdot (1+\alpha) \sum_{i \in N}  d(i,o_i) \\
    &\leq (2\rho k + 1) \cdot \cost(O).
\end{align*}
The proof is complete.
\end{proof}

We now show a lower bound of $2k+1$ on the distortion of any ordinal-information algorithm for one-sided matching and clustering. For one-sided matching, this implies that $2k+1$ is the best possible overall distortion, while for clustering, this is a lower bound on the distortion even when unlimited computing resources are available. 

\begin{theorem}
For the Max-of-Sum objective, the distortion of any ordinal-information matching algorithm is at least $2k+1$.
\end{theorem}

\begin{proof}
Consider the following instance:
\begin{itemize}
    \item There are $k^2$ facilities located at $0$ and $k$ facilities located at $1$.
    \item There are $k(k+1)$ agents with the same ordinal preference over the locations: $0 \succ 1$. That is, the agents prefer all the facilities at $0$ over all the facilities at $1$; ties for facilities at the same location are broken arbitrarily.
\end{itemize}
Let $W$ be the matching computed by the algorithm, and denote by $S$ the set of $k$ agents that are assigned to the facilities at $1$.
Consider the following partition of the agents into $k$ groups and a consistent metric space:
\begin{itemize}
    \item The first group is of size $k+1$. It consists of the agents in $S$ that are assigned to the facilities at $1$ and are located at $0$, and one other agent who is assigned to a facility at $0$ and is located at $1/2$.
    \item Each of the remaining $k-1$ groups also has size $k+1$. It consists of $k$ agents that are located at $0$ and one agent that is located at $1/2$. All of these agents are assigned to facilities at $0$. 
\end{itemize}
The Max-of-Sum cost of $W$ is realized by the first group. Each agent in $S$ has a cost of 
$1+ \frac{\alpha}{k}\left(k-1+\frac12\right)$ and the agent at $1/2$ has a cost of $\frac12 + \frac{\alpha}{k}k$. 
Hence,  
\begin{align*}
    \cost(W) = k \bigg( 1+ \frac{\alpha}{k}\left(k-\frac12\right) \bigg) + \frac12 + \alpha = \left(k+\frac12\right) (1+\alpha).
\end{align*}
We can create an optimal matching such that all the $k^2$ agents that are located at $0$ (there are $k$ such agents per group) are assigned to the $k^2$ facilities at $0$, and the remaining $k$ agents that are located at $1/2$ (one per group) are assigned to the $k$ facilities at $1$. In each group, there is an agent with cost $1/2$ and $k$ agents with cost $\frac{\alpha}{k} \cdot \frac12$. Hence, the optimal cost is  
\begin{align*}
\cost(O) = \frac12 + k \cdot \frac{\alpha}{2k} = \frac12 (1+\alpha),
\end{align*}
and the distortion is $2k+1$.
\end{proof}

\begin{theorem}
For the Max-of-Sum objective, the distortion of any ordinal-information clustering algorithm is at least $2k+1$.
\end{theorem}

\begin{proof}
Consider the following instance:
\begin{itemize}
    \item There are three facilities at $0$, $1$ and $2$.
    \item There are $k$ agents with ranking $0 \succ 1 \succ 2$.
    \item There are $k$ agents with ranking $2 \succ 1 \succ 0$.
    \item There are $(k-1)k$ agents with $1$ as their most-preferred location. These agents are partitioned into sets $S_0$ and $S_2$ of equal size, such that the agents in $S_0$ prefer $0$ over $2$ (i.e., their ranking is $1 \succ 0 \succ 2$), and the agents in $S_2$ prefer $2$ over $0$ (i.e., their ranking is $1 \succ 2 \succ 0$). 
\end{itemize}
We switch between the possible locations of the cluster centers chosen by the algorithm; recall that each agent is assigned to its closest cluster center.

\medskip
\noindent
{\bf Case 1:} The cluster centers are $(0,2)$.
The $k$ groups and the consistent metric space might be as follows:
\begin{itemize}
\item The first group consists of the $k$ agents with ranking $0 \succ 1 \succ 2$ that are positioned at $0$, and one agent with ranking $2 \succ 1 \succ 0$ that is positioned at $3/2$.

\item Each of the remaining $k-1$ groups consists of $k$ agents with $1$ as their most-preferred location that are positioned at $1$, and one agent with ranking $2 \succ 1 \succ 0$ that is positioned at $3/2$.
\end{itemize}
The Max-of-Sum cost of the computed clustering $W$ is realized by any of the last $k-1$ groups.
In each such group, the agents at $1$ have cost $1+\frac{\alpha}{k}(k-1+1/2)=1+\frac{k-1/2}{k}\alpha$ (no matter whether they are assigned to $0$ or $2$), and the agent at $3/2$ has cost $1/2 + \frac{\alpha}{k}k = 1/2 + \alpha$ (as it is assigned to $2$). 
Hence,
\begin{align*}
    \cost(W) = k \left(1+\frac{k-1/2}{k}\alpha\right) + 1/2 + \alpha = \left(k+\frac12\right)(1+\alpha).
\end{align*}
By choosing $(0,1)$ as the cluster centers, we can compute a clustering $O$ such that, for each group, 
the agent at $3/2$ has cost $1/2$ (as it is assigned to $1$) and 
each of the remaining $k$ agents (either at $0$ in the first group, or at $1$ in any of the last $k-1$ groups) 
has cost $\frac{\alpha}{k}\cdot \frac12$, for a total of 
\begin{align*}
    \cost(O) = \frac12 + k\cdot \frac{\alpha}{2k} = \frac12 (1+\alpha).
\end{align*}
Hence, the distortion is at least $2k+1$ in this case.

\medskip
\noindent
{\bf Case 2:} The cluster centers are $(0,1)$ (the case $(1,2)$ is symmetric).
The $k$ groups and the consistent metric space might be as follows:
\begin{itemize}
    \item The first group consists of the $k$ agents with ranking $2 \succ 1 \succ 0$ that are positioned at $2$, and one agent with ranking $0 \succ 1\succ 2$ that is positioned at $1/2$.
    
    \item Each of the remaining $k-1$ groups consists of $k$ agents with $1$ as their most-preferred location that are positioned at $1$, and one agent with ranking $0 \succ 1 \succ 2$ that is positioned at $1/2$.
\end{itemize}
The Max-of-Sum cost of the computed clustering $W$ is realized by the first group. 
Each of the $k$ agents at $2$ has cost $1+\frac{\alpha}{k}(k-1+1/2) = 1+\frac{k-1/2}{k}\alpha$ (as they are assigned to $1$), 
and the agent at $1/2$ has cost $1/2 + \frac{\alpha}{k}k = 1/2 + \alpha$ (no matter whether it is assigned to $0$ or $1$). 
So,
\begin{align*}
    \cost(W) = k \left(1+\frac{k-1/2}{k}\alpha\right) + 1/2 + \alpha = \left(k+\frac12\right)(1+\alpha).
\end{align*}
By choosing $(1,2)$ as the cluster centers, we can compute a clustering $O$ such that, for each group, the agent at $1/2$ has cost $1/2$ (as it is assigned to $1$) and each of the remaining $k$ agents (either at $2$ in the first group, or at $1$ in any of the remaining $k-1$ groups) has cost $\frac{\alpha}{k}\cdot \frac12$, for a total of 
\begin{align*}
    \cost(O) = \frac12 + k\cdot\frac{\alpha}{2k} = \frac12 (1+\alpha).
\end{align*}
Hence, the distortion is again at least $2k+1$.
\end{proof}

\section{Sum-of-Max} \label{sec:sum-of-max}
In this section, we consider the second objective, Sum-of-Max. For this objective, we will need the following technical lemma, which is similar in spirit to Lemma~\ref{lem:group-cost} but involves agents that maximize their assignment distance within their groups. We first observe that every such agent also has maximum individual cost in its group. For simplicity, we assume that the $\argmax$ operator returns a single arbitrary element among those that satisfy the condition.

\begin{lemma} \label{lem:separate-group-cost}
For any assignment $X$, let
\[
S_X = \bigcup_{g \in G} \argmax_{i \in g} d(i,x_i)
\]
be a $k$-sized set containing one maximum-assignment-distance agent from each group. Every agent in $S_X$ has maximum individual cost within its group, and
\begin{align*}
    \sum_{i \in S_X} \cost_i(X) \leq (1+\alpha) \sum_{i \in S_X} d(i,x_i).
\end{align*}
\end{lemma}

\begin{proof}
Fix a group $g$ and write $D_g(X)=\sum_{j\in g}d(j,x_j)$. For every agent $i\in g$,
\begin{align*}
\cost_i(X)
&=d(i,x_i)+\frac{\alpha}{n_g-1}\bigl(D_g(X)-d(i,x_i)\bigr)\\
&=\left(1-\frac{\alpha}{n_g-1}\right)d(i,x_i)
  +\frac{\alpha}{n_g-1}D_g(X).
\end{align*}
The second term is the same for every member of $g$, while
$1-\frac{\alpha}{n_g-1}\geq 0$ because $\alpha\in[0,1]$ and $n_g\geq2$.
Consequently, every agent that maximizes $d(i,x_i)$ in $g$ also maximizes
$\cost_i(X)$ in $g$.

Now consider any agent $i\in S_X$ belonging to group $g_i$. Since
$d(j,x_j)\leq d(i,x_i)$ for every $j\in g_i\setminus\{i\}$, we obtain
\begin{align*}
   \sum_{i \in S_X} \cost_i(X)
    &=  \sum_{i\in S_X} \bigg( d(i,x_i) + \frac{\alpha}{n_{g_i}-1} \sum_{j \in g_i\setminus\{i\}} d(j,x_j) \bigg) \\
    &\leq \sum_{i\in S_X} \bigg( d(i,x_i) + \frac{\alpha}{n_{g_i}-1} \cdot (n_{g_i}-1) \cdot d(i,x_i) \bigg) \\
    &= (1+\alpha) \cdot \sum_{i\in S_X} d(i,x_i),
\end{align*}
as desired.
\end{proof}

\subsection{Full Information Algorithms} \label{sec:sum-of-max:full}
As in Section~\ref{sec:max-of-sum:full}, we have complete knowledge of the metric space. Given the structure of the Sum-of-Max objective, it is natural to consider $\rho$-approximation algorithms for the {\em MAX problem} of minimizing the maximum distance among all agents from their assigned facilities.\footnote{For one-sided matching, the MAX problem is the bottleneck matching problem~\citep{Derigs1979bottleneck}, while for clustering, it is equivalent to $\lambda$-center.} In other words, our algorithm outputs an assignment $W$ which, for any assignment $X$, satisfies the inequality 
\begin{align*}
    \max_{i \in N} d(i,w_i) \leq \rho \cdot \max_{i \in N} d(i,x_i). 
\end{align*}
We show that such an assignment achieves a distortion of $\rho k (1+\alpha)$ for the Sum-of-Max objective. 

\begin{theorem} \label{thm:sum-of-max:full:upper}
Any $\rho$-approximate algorithm for the MAX problem achieves a distortion of at most $(1+\alpha) \rho k $ in terms of the Sum-of-Max objective.
\end{theorem}

\begin{proof}
Let $W$ be the assignment computed by the algorithm, and let $O$ be an optimal assignment for the Sum-of-Max objective. Since the algorithm is $\rho$-approximate for the MAX problem, we have
\begin{align} \label{eq:sum-of-max:full:algorithm}
    \max_{i \in N} d(i,w_i) \leq \rho \cdot \max_{i \in N} d(i,o_i). 
\end{align}
Let $S_O$ be the set defined in Lemma~\ref{lem:separate-group-cost} for the assignment $O$. Since its representative from each group also has maximum individual cost in that group,
\begin{align} \label{eq:sum-of-max:full:optimal}
    \cost(O) = \sum_{i \in S_O} \cost_{i}(O)
    \geq \max_{i \in S_O} d(i,o_i)
    = \max_{i \in N} d(i,o_i).
\end{align}
Now, let $S_W$ be the set defined in Lemma~\ref{lem:separate-group-cost} for the assignment $W$. By Lemma~\ref{lem:separate-group-cost}, \eqref{eq:sum-of-max:full:algorithm} and \eqref{eq:sum-of-max:full:optimal}, we have
\begin{align*}
    \cost(W)
    &= \sum_{i\in S_W}\cost_i(W) \\
    &\leq (1+\alpha) \cdot \sum_{i\in S_W} d(i,w_i) \\
    &\leq (1+\alpha) k \cdot \max_{i\in S_W} d(i,w_i) \\
    &\leq (1+\alpha) k \cdot \max_{i\in N} d(i,w_i) \\
    &\leq (1+\alpha) \rho k \cdot \max_{i\in N} d(i,o_i) \\
    &\leq (1+\alpha) \rho k \cdot \cost(O).
\end{align*}
The proof is complete.
\end{proof}

We now show a lower bound of $(k-1)(1+\alpha)+1$ for one-sided matching and clustering. Together with Theorem~\ref{thm:sum-of-max:full:upper}, this lower bound is asymptotically tight in $k$ whenever the underlying MAX problem can be solved optimally.

\begin{theorem}
For the Sum-of-Max objective, the distortion of any full-information matching algorithm is at least $(1+\alpha)(k-1)+1$.
\end{theorem}

\begin{proof}
Let $x \geq 2k-1$ be a parameter and consider the following instance on a line metric:
\begin{itemize}
    \item There are $n=x+2(k-1)+1=x+2k-1$ agents located at $0$.
    \item There are $x$ facilities located at $0$ and $2k-1$ facilities located at $1$.
\end{itemize}
Let $W$ be the matching computed by the algorithm, and denote by $S$ the set of $2k-1$ agents that are assigned to the facilities at $1$.
The groups might be as follows:
\begin{itemize}
    \item The first group contains the $x$ agents that are assigned to the facilities at $0$ and one agent from $S$.
    \item Each of the remaining $k-1$ groups consists of two agents from $S$. 
\end{itemize}
The maximum cost in the first group is $1$, and the maximum cost in each of the remaining $k-1$ groups is $1+\alpha$. 
Hence, the Sum-of-Max cost of $W$ is
\begin{align*}
    \cost(W) = (k-1)(1+\alpha)+1.
\end{align*}
We can create a matching $O$ such that a subset of $2k-1$ agents in the first group are assigned to the facilities at $1$, while all remaining agents are assigned to the facilities at $0$. Then, the optimal Sum-of-Max cost is realized by the first group and is equal to
\begin{align*}
    \cost(O) = 1+\frac{2k-2}{x}\cdot \alpha.
\end{align*}
By making $x$ arbitrarily large, we obtain a lower bound of $(1+\alpha)(k-1) + 1$ on the distortion.
\end{proof}

\begin{theorem}
For the Sum-of-Max objective, the distortion of any full-information clustering algorithm is at least $(1+\alpha)(k-1)+1$.
\end{theorem}

\begin{proof}
Let $x$ be a parameter and consider the following instance on a line metric:
\begin{itemize}
    \item There are three facilities at $0$, $1$ and $2$.
    \item There are $2k+1$ agents at $0$, $x+2(k-1)$ agents at $1$, and $2k+1$ agents at $2$.
\end{itemize}
We consider the possible pairs of cluster centers chosen by the algorithm.

\medskip
\noindent
{\bf Case 1:} The cluster centers are $(0,2)$.
The $k$ groups might be as follows:
\begin{itemize}
\item The first group consists of the $2k+1$ agents at $0$, the $2k+1$ agents at $2$, and $x$ agents at $1$.
\item Each of the remaining $k-1$ groups consists of $2$ agents at $1$.
\end{itemize}
The maximum cost in the first group is $1+\frac{\alpha}{x+4k+1}(x-1)$ (realized by any agent at $1$), and the maximum cost in each of the last $k-1$ groups is $1+\alpha$. Hence, the Sum-of-Max cost of the computed clustering $W$ is 
\begin{align*}
    \cost(W) = (k-1)(1+\alpha) + 1 + \frac{x-1}{x+4k+1}\cdot \alpha.
\end{align*}
By choosing $(0,1)$, the maximum cost in the first group is $1+\frac{2k}{x+4k+1}\alpha$ (realized by any agent at $2$), and the maximum cost in each of the last $k-1$ groups is $0$. Hence, the Sum-of-Max cost of the optimal clustering $O$ is 
\begin{align*}
    \cost(O) = 1+\frac{2k}{x+4k+1}\cdot \alpha.
\end{align*}
By making $x$ arbitrarily large, the distortion is at least $(1+\alpha)k$. 

\medskip
\noindent
{\bf Case 2:} The cluster centers are $(0,1)$ (the case of $(1,2)$ is symmetric).
The $k$ groups might be as follows:
\begin{itemize}
    \item The first group consists of the $2k+1$ agents at $0$, the $x+2(k-1)$ agents at $1$, and three agents at $2$.
    \item Each of the remaining $k-1$ groups consists of two agents at $2$.
\end{itemize}
This is a partition of all agents, since the first group contains three of the $2k+1$ agents at $2$, while the remaining groups contain the other $2(k-1)$ such agents. Under the centers $(0,1)$, the maximum cost in the first group is
\begin{align*}
    1+\frac{2\alpha}{x+4k+1},
\end{align*}
realized by any of its three agents at $2$, and the maximum cost in each of the last $k-1$ groups is $1+\alpha$. Hence, the Sum-of-Max cost of the computed clustering $W$ is
\begin{align*}
    \cost(W) = (k-1)(1+\alpha)+1+\frac{2\alpha}{x+4k+1}.
\end{align*}
By choosing $(1,2)$, the maximum cost in the first group is
\begin{align*}
    1+\frac{2k\alpha}{x+4k+1},
\end{align*}
realized by any agent at $0$, and the maximum cost in each of the last $k-1$ groups is $0$. Thus, for this clustering $O$,
\begin{align*}
    \cost(O) = 1+\frac{2k\alpha}{x+4k+1}.
\end{align*}
By making $x$ arbitrarily large, the distortion is at least $(1+\alpha)(k-1)+1$.
\end{proof}

\subsection{Ordinal Information Algorithms} \label{sec:sum-of-max:ordinal}
For ordinal-information algorithms, as in Section~\ref{sec:max-of-sum:ordinal}, we form a proxy instance in which each agent $i$ is located at its top-ranked facility $\top(i)$ and apply a $\rho$-approximate algorithm for the MAX problem. Let $\widehat{W}=(\widehat{w}_i)_{i\in N}$ denote the assignment computed in the proxy instance. For clustering, the proxy algorithm selects a set $C_W$ of cluster centers, $\widehat{w}_i$ is a closest center in $C_W$ to $\top(i)$, and the actual ordinal algorithm assigns $i$ to its highest-ranked center $w_i$ in $C_W$. For all other facility assignment problems, we set $w_i=\widehat{w}_i$. Hence, as in Inequality~\eqref{eq:minitop-sum:transfer},
\begin{align} \label{eq:minitop-max:transfer}
    d(i,w_i) \leq d(i,\top(i)) + d(\top(i),\widehat{w}_i).
\end{align}
Our algorithm, named {\sc $\rho$-MiniTopMAX}, outputs $W=(w_i)_{i\in N}$, while the proxy assignment satisfies, for every feasible comparison assignment $X=(x_i)_{i\in N}$,
\begin{align} \label{eq:minitop-max:proxy}
   \max_{i \in N} d(\top(i),\widehat{w}_i) \leq \rho \cdot \max_{i \in N} d(\top(i),x_i).
\end{align}
For clustering, this inequality again follows by comparing the selected center set with the center set inducing $X$ and assigning each proxy agent to its closest selected center.

\begin{theorem} \label{thm:sum-of-max:ordinal:upper}
Given any $\rho$-approximate algorithm for the MAX problem, the distortion of {\sc $\rho$-MiniTopMAX} is at most $(1+\alpha)(2\rho k+1)$ in terms of the Sum-of-Max objective. 
\end{theorem}

\begin{proof}
Let $W$ be the assignment computed by the algorithm, and let $O$ be an optimal assignment for the Sum-of-Max objective. 
In addition, let $S_W$ and $S_O$ be the sets defined in Lemma~\ref{lem:separate-group-cost} for $W$ and $O$, respectively. Their selected agents maximize both assignment distance and individual cost within their groups. For the optimal cost, we again use Inequality~\eqref{eq:sum-of-max:full:optimal}, which states that
\begin{align*}
    \cost(O) = \sum_{i \in S_O} \cost_i(O) \geq \max_{i \in N} d(i,o_i).
\end{align*}
Using Lemma~\ref{lem:separate-group-cost}, the transfer bound in Inequality~\eqref{eq:minitop-max:transfer}, the fact that $d(i,\top(i))\leq d(i,o_i)$, the $\rho$-approximation guarantee of the proxy assignment in Inequality~\eqref{eq:minitop-max:proxy}, and the above lower bound on the optimal cost (Inequality~\eqref{eq:sum-of-max:full:optimal}), we have
\begin{align*}
    \cost(W) 
     &= \sum_{i \in S_W} \cost_i(W) \\
     &\leq (1+\alpha) \sum_{i \in S_W} d(i,w_i) \\
     &\leq (1+\alpha) \sum_{i \in S_W} d(i,\top(i)) + (1+\alpha) \sum_{i \in S_W} d(\top(i),\widehat{w}_i) \\
     &\leq (1+\alpha) \sum_{i \in S_W} d(i,o_i) + (1+\alpha)k \cdot \max_{i \in S_W} d(\top(i),\widehat{w}_i) \\
     &\leq (1+\alpha) \sum_{i \in S_O} d(i,o_i) + (1+\alpha)k \cdot \max_{i \in N} d(\top(i),\widehat{w}_i)  \\
     &\leq (1+\alpha) \cdot \cost(O) +  (1+\alpha)\rho k \cdot \max_{i \in N} d(\top(i),o_i)  \\
     &\leq (1+\alpha) \cost(O) +  2(1+\alpha)\rho k \cdot \max_{i \in N} d(i,o_i)  \\
     &\leq (1+\alpha) (2\rho k+1) \cdot \cost(O).
\end{align*}
The proof is complete.
\end{proof}

Finally, we show informational lower bounds for one-sided matching and clustering that are asymptotically tight in $k$. 

\begin{theorem}
For the Sum-of-Max objective, the distortion of any ordinal-information matching algorithm is at least $2(k-1)(1+\alpha)+3$.
\end{theorem}

\begin{proof}
Let $x \geq 2k-1$ be a parameter and consider the following instance on a line metric:
\begin{itemize}
    \item There are $x$ facilities located at $0$ and $2k-1$ facilities located at $1$.
    \item There are $n=x+2(k-1)+1=x+2k-1$ agents with the same ordinal preferences such that they prefer the facilities at $0$ over the facilities at $1$.
\end{itemize}
Let $W$ be the matching computed by the algorithm, and denote by $S$ the set of $2k-1$ agents that are assigned to the facilities at $1$. 
Consider the following partition of the agents into groups and a consistent metric space:
\begin{itemize}
    \item The first group is of size $x+1$. It includes $x-(2k-1)$ agents that are assigned to facilities at $0$ and are located at $0$, $2k-1$ agents that are assigned to facilities at $0$ and are located at $1/2$, and one agent that is assigned to a facility at $1$ and is located at $-1/2$. 

    \item Each of the remaining $k-1$ groups is of size $2$ and consists of two agents that are assigned to facilities at $1$ and are located at $0$. 
\end{itemize}
The maximum cost in the first group is realized by the agent at $-1/2$ and is equal to $\frac{3}{2} + \frac{\alpha}{x}\cdot\frac{2k-1}{2}$. The maximum cost in each of the remaining $k-1$ groups is $1+\alpha$. Hence, the Sum-of-Max cost of $W$ is 
\begin{align*}
    \cost(W) = (k-1)(1+\alpha) + \frac{3}{2} + \frac{\alpha}{x}\cdot\frac{2k-1}{2}.
\end{align*}
We can create a matching $O$ that assigns the $2k-1$ agents in the first group that are located at $1/2$ to the facilities at $1$, and all other agents to the facilities at $0$. Then, the maximum cost in the first group is $\frac{1}{2} + \frac{\alpha}{x}\cdot\frac{2k-1}{2}$ (realized by any agent located at $1/2$ or the agent located at $-1/2$), while the maximum cost in any other group is $0$. So, the optimal Sum-of-Max cost is  
\begin{align*}
    \cost(O) = \frac{1}{2} + \frac{\alpha}{x}\cdot\frac{2k-1}{2}.
\end{align*}
By making $x$ arbitrarily large, we obtain a lower bound of $2(k-1)(1+\alpha)+3$.
\end{proof}

\begin{theorem}
For the Sum-of-Max objective, the distortion of any ordinal-information clustering algorithm is at least 
$2(k-1)(1+\alpha)+2$.
\end{theorem}

\begin{proof}
Let $x$ be an even parameter and consider the following instance on a line metric:
\begin{itemize}
    \item There are three facilities at $0$, $1$ and $2$.
    \item There are $2k+1$ agents with ranking $0 \succ 1 \succ 2$.
    \item There are $2k+1$ agents with ranking $2 \succ 1 \succ 0$.
    \item There are $x+2(k-1)$ agents with $1$ as their most-preferred location. These agents are partitioned into sets $S_0$ and $S_2$ of equal size, such that the agents in $S_0$ prefer $0$ over $2$ (i.e., their ranking is $1 \succ 0 \succ 2$), and the agents in $S_2$ prefer $2$ over $0$ (i.e., their ranking is $1 \succ 2 \succ 0$).  
\end{itemize}
We switch between the possible locations of the cluster centers.

\medskip
\noindent
{\bf Case 1:} The cluster centers are $(0,2)$.
The $k$ groups and the consistent metric space might be as follows:
\begin{itemize}
\item The first group consists of the $2k+1$ agents with ranking $0 \succ 1 \succ 2$ that are positioned at $0$, 
the $2k+1$ agents with ranking $2 \succ 1 \succ 0$ that are positioned at $3/2$, and $x$ agents with $1$ as their most-preferred location that are positioned at $1$.

\item Each of the remaining $k-1$ groups consists of $2$ agents with $1$ as their most-preferred location that are positioned at $1$.
\end{itemize}
The maximum cost in the first group is $1+\frac{\alpha}{x+4k+1}(x-1 + (2k+1)\frac12)$ (realized by any agent at $1$), and the maximum cost in each of the last $k-1$ groups is $1+\alpha$. Hence, the Sum-of-Max cost of the computed clustering $W$ is 
\begin{align*}
    \cost(W) = (k-1)(1+\alpha) + 1 + \frac{x+k-1/2}{x+4k+1}\cdot \alpha.
\end{align*}
By choosing $(0,1)$, the maximum cost in the first group is $\frac12+\frac{\alpha}{x+4k+1}\cdot 2k\frac12$ (realized by any agent at $3/2$), and the maximum cost in each of the last $k-1$ groups is $0$. Hence, the Sum-of-Max cost of this optimal clustering $O$ is 
\begin{align*}
    \cost(O) = \frac12 + \frac{k}{x+4k+1}\alpha.
\end{align*}
By making $x$ arbitrarily large, the distortion is at least $2(k-1)(1+\alpha)+2$. 

\medskip
\noindent
{\bf Case 2:} The cluster centers are $(0,1)$ (the case $(1,2)$ is symmetric).
The $k$ groups might be as follows:
\begin{itemize}
    \item The first group consists of the $2k+1$ agents with ranking $0 \succ 1 \succ 2$ that are positioned at $1/2$, 
    the $x+2(k-1)$ agents with $1$ as their most-preferred location that are positioned at $1$, and $3$ agents with ranking $2 \succ 1 \succ 0$ that are positioned at $2$.
    \item Each of the remaining $k-1$ groups consists of two agents at $2$.
\end{itemize}
The maximum cost in the first group is $1 + \frac{\alpha}{x+4k+1}\left(2+(2k+1)\frac12\right)$ (realized by any agent at $2$), and the maximum cost in each of the last $k-1$ groups is $1+\alpha$. Hence, the Sum-of-Max cost of the computed clustering $W$ is
\begin{align*}
    \cost(W) = (k-1)(1+\alpha) + 1 + \frac{k+5/2}{x+4k+1} \alpha.
\end{align*}
By choosing $(1,2)$, the maximum cost in the first group is $1/2+\frac{\alpha}{x+4k+1}\cdot 2k\frac12$ (realized by any agent at $1/2$), and the maximum cost in each of the other $k-1$ groups is $0$. Hence, the optimal Sum-of-Max cost is
\begin{align*}
    \cost(O) = \frac12 + \frac{k}{x+4k+1} \alpha.
\end{align*}
By making $x$ arbitrarily large, the distortion is at least $2(k-1)(1+\alpha)+2$. 
\end{proof}

\section{Conclusion}\label{sec:conclusion}
In this paper, we investigated the group-fair distortion of metric facility assignment problems, paying particular attention to one-sided matching and clustering. We derived bounds on the distortion of algorithms that are given the exact distances between facilities and different types of information about the distances between agents and facilities. For one-sided matching and clustering, these bounds are exact for Max-of-Sum and asymptotically tight in $k$ for Sum-of-Max. That is, we considered full-information algorithms, which have complete access to the underlying metric space and whose inefficiency is due to the unknown group structure, as well as ordinal-information algorithms, which have access only to the rankings of agents over facilities and incur additional inefficiency due to this limited information.

There are many interesting directions for future work, such as considering alternative group-fair objectives beyond the standard Max-of-Sum and Sum-of-Max objective functions; for example, there are natural generalizations that combine different $p$-norms or $k$-centrum costs within and across groups. In addition, it would be interesting to study the distortion of algorithms with access to different types of information about the metric space, such as algorithms that have limited (e.g., ordinal) information about the distances between facilities, have approval information about the preferences of the agents over the facilities, or can make queries to obtain partial access to the metric space. Another promising direction is the analysis of randomized and learning-augmented algorithms~\citep{BergerFGT24,Filos-Ratsikas025}, which, for example, may be given predictions about the unknown preferences of the agents or about the group structure.

\bibliographystyle{plainnat}
\bibliography{references}

\end{document}